\documentclass[10pt, journal]{IEEEtran}  

\IEEEoverridecommandlockouts     

\usepackage{amsmath,amsthm,amstext,amsfonts,amssymb,mathrsfs,dsfont}
\usepackage{graphicx,subfigure,algorithm,algorithmic}
\usepackage{tikz}
\usepackage{setspace}
\usepackage{array}
\usepackage{booktabs}
\usepackage{bbm}
\usepackage[sort,compress]{cite}

\newtheorem{theorem}{Theorem}
\newtheorem{lemma}{Lemma}

\newtheorem{remark}{Remark}

\newtheorem{definition}{Definition}

\title{A Minimal Incentive-based Demand Response Program With Self Reported Baseline Mechanism}
\author{
Deepan Muthirayan$^a$,
Enrique Baeyens$^b$,
Pratyush Chakraborty$^c$,\\
Kameshwar Poolla$^d$ and
Pramod P. Khargonekar$^a$
\thanks{This research is supported by the National Science Foundation under
grants EAGER-1549945, CPS-1646612, CNS-1723856 and by the National Research
Foundation of Singapore under a grant to the Berkeley Alliance for Research
in Singapore}
\thanks{$^a$ Department of Electrical Engineering and Computer Science,
        University of California, Irvine, CA, USA}%
\thanks{$^b$ Instituto de las Tecnolog\'{\i}as Avanzadas de la Producci\'on,
Universidad de Valladolid, Valladolid, Spain
        }%
\thanks{$^c$ Department of Physics and Astronomy,
        Northwestern University, IL, USA
        }
\thanks{$^d$ Department of Electrical Engineering and Computer Science,
        University of California, Berkeley, CA, USA}%
}

\date{}

\begin{document}

\maketitle

\begin{abstract}


In this paper, we propose a novel incentive based Demand Response (DR) program
with a self reported baseline mechanism. The System Operator (SO) managing the DR program recruits consumers or aggregators of DR resources. The recruited consumers are required to only report their baseline, which is the minimal information necessary for any DR program. During a DR event, a set of  consumers, from this pool of recruited consumers, are randomly selected. The consumers are selected such that the required load reduction is delivered. The selected consumers, who reduce their load, are rewarded for their services and other recruited consumers, who deviate from their reported baseline, are penalized. The randomization in selection and penalty ensure that the \emph{baseline inflation} is controlled. We also justify that the selection probability can be simultaneously used to control SO's cost. This allows the SO to design the mechanism such that its cost is almost optimal when there are no recruitment costs or atleast significantly reduced otherwise. Finally, we also show that the proposed method of self-reported baseline outperforms other baseline estimation methods commonly used in practice.
\end{abstract}

\section{Introduction}

Demand Response (DR) programs \cite{Albadi2008} are potentially powerful tools to modulate
the demand for electricity in a wide variety of situations. For example, at certain times such as mid-afternoons on hot summer
days, the supply of additional electric power is scarce and expensive. At these
times, it is more cost-effective to reduce demand than to increase supply to
maintain power balance. Another scenario is a grid with high renewable
penetration. Here, DR promises to be a better alternative compared to other expensive and
polluting reserves to balance the variability in renewable generation. Realizing its potential, the 2005 Energy Policy Act
provided the Congressional mandate to promote DR in organized
wholesale electricity markets. The FERC
order 745 \cite{federal2011demand} met this mandate by prescribing that demand
response resource owners should be allowed to offer their demand reduction as
if it were a supply resource rather than a bid to reduce demand so that the
market operates fairly.

Dynamic pricing based DR programs \cite{Joskow2012, chakraborty2017distributed} can ideally
achieve market efficiency, but they require more complex metering and communication infrastructure to achieve this which raises their implementation costs \cite{mathieu2013residential, borenstein2002dynamic}.
Furthermore, consumers may not be responsive to dynamic pricing
\cite{faruqui2011dynamic}. Alternatively, consumers could be signaled to reduce
consumption and paid for their load reductions. Such schemes are referred to as
Incentive-based DR programs or Demand Reduction programs. There are three key
components of any incentive-based DR program: (a) a baseline against which
demand reduction is measured, (b) a payment scheme for agents who reduce their
consumption from the baseline, and (c) various contractual clauses such as
limits on the frequency of DR events or penalties for nonconforming agents.

Thus, incentive-based DR programs require an established baseline against which
consumer's load reduction is measured. The baseline is an estimate of the
consumption when the consumer is not participating in the DR
program.
For example, the California Independent System Operator (CAISO) uses the
average of the consumption on the ten most recent non-event days as the
baseline estimate~\cite{CaisoDR2017}.
The CAISO method also uses a morning adjustment factor to
account for any variability in consumption pattern during the day of the DR
event from the past. Current methods to establish the baseline raise several concerns. One major concern is that the consumers have an incentive to artificially inflate their baseline to increase
their profits \cite{chao2010price,wolak2007residential,chao2013incentive, vuelvas2017rational}.
Cases have been reported where the participants artificially inflated their
baseline for increasing payments \cite{gaming-examples}. Fairness can also
be a concern. Consider, for instance, an agent who happens to be on vacation
during a DR event and receives a payment for load reduction without suffering
any hardship. This can be perceived as unfair by other agents who deliberately
curtail their consumption and suffer some disutility.

\subsection{Our Contribution}

We propose a setting where the System Operator (SO) recruits DR
providers as an alternate resource to balance supply and demand during high
price periods. The providers could be either individual consumers or
aggregators of DR services. We also assume that the SO has access to market
outcomes, which is a reasonable assumption. The objective of the SO is to minimize cost when energy purchase from the wholesale energy market is expensive. This usually happens during
peak load scenarios, when the market price exceeds a \emph{threshold market
clearing} (TMC) price. The TMC price is the price above which it is profitable for the SO to call the recruited DR providers or consumers to provide load reduction.

The main aspects of the DR mechanism we propose are: (i) self-reported baseline (ii) randomized selection of consumers, and (iii) penalty for uninstructed deviations. In this mechanism, the consumers are required to self-report their baselines and are paid at a pre-determined reward for every unit of reduction they provide. A large group of consumers is recruited so that the necessary
load reduction is delivered reliably. When a DR event occurs, consumers are selected randomly from this pool of recruited consumers to provide the required service. The load reduction is measured by the difference between the self-reported baseline and the measured consumption. The consumers signaled to reduce are paid in proportion to the measured reduction and the prescribed reward. The consumers who are not called are penalized for uninstructed deviation from the baseline. This penalty and randomized selection controls baseline inflation. The proposed DR program requires only baseline information from the individual consumers, which makes it minimal in terms of the information it elicits from the consumers.

In this paper, we characterize baseline inflation for a quadratic utility function with
uncertain consumption and a quadratic penalty function with and without a
deadband. The deadband in the penalty function is required to achieve
individual rationality. Using this characterization, we show that the proposed mechanism
controls baseline inflation. We also justify that by choosing an appropriate calling probability, which depends on the recruitment cost, the SO can signficantly reduce its costs. Finally, we show that the self-reported baseline establishes a better estimate of the mean baseline when compared to conventional methods such as  the CAISO's m/m method~\cite{CaisoDR2017}. Since the excess payments made to the consumers are proportional to the baseline used in a DR program, this establishes that the self-reported baseline approach is more cost effective than the CAISO approach.

Two concerns can arise with the self-reported baseline idea. One is the fatigue in reporting a baseline and the other is the lack of knowledge of one's own baseline. Notwithstanding, self-reported baseline is still a viable method. This is because, firstly, the proposed mechanism is for peak load scenarios which are rare events. Secondly, we expect a energy management system to manage the load consumption pattern of a consumer in the future. Given the consumer's preferences, this energy management system should have the capability to estimate the baseline and report it to the operator or the load serving entity. In addition to all of the above aspects, the self-reported baseline DR mechanism can also avoid bias and inflation in its estimate of the baseline.

\subsection{Related Work}

There exists substantial literature on baseline estimation methods \cite{coughlin2008estimating, grimm2008evaluating, mathieu2011quantifying, wijaya2014bias, nolan2015challenges, weng2015probabilistic, zhang2016cluster, nolan2015challenges, zhou2016forecast, hatton2016statistical, wang2018synchronous}. These can be broadly classified into three classes: (a) averaging, (b) regression, and (c) control group methods.

\emph{Averaging methods} determine baselines by averaging the consumption on
past days that are similar (\emph{e.g.,} in weather conditions) to the event
day.
A detailed comparison of  different averaging methods is offered in
\cite{coughlin2008estimating, grimm2008evaluating, wijaya2014bias}. Averaging
methods are simple but they suffer from estimation biases
\cite{wijaya2014bias, weng2015probabilistic, nolan2015challenges}, and require a significant
amount of data, especially for residential DR programs \cite{nolan2015challenges}.

\emph{Regression methods} estimate a load prediction model based on historical data which is then
used to predict the baseline \cite{zhou2016forecast, mathieu2011quantifying}.
They can potentially overcome biases incurred by averaging methods
\cite{mathieu2011examining, nolan2015challenges}. But they often require considerable historical data
for acceptable accuracy, and the models may not capture the complex behavior of
individual consumers.

\emph{Control group methods} have been suggested to have better accuracy than averaging or regression methods and do not require large amounts of historical data \cite{hatton2016statistical}. However these methods require the SO to recruit an additional set of consumers and also install additional metering infrastructure for these consumers. In addition, prior data based analysis, to identify the most appropriate control group, might be required depending on the control group method deployed. This raises their costs of implementation \cite{hatton2016statistical}. We also show later that the adverse incentives to inflate still persists in this method. Compared to all the above methods, the proposed {\it self-reported baseline} avoids all of these issues, \emph{i.e.} (i) bias and inflation, (ii) need for historical data, and (iii) high implementation cost.

In order to avoid baseline estimation, in a previous
work~\cite{muthirayan2017mechanism}, we addressed the DR problem as a mechanism design
problem. The setting considered in [23] has an aggregator and an Utility or SO. The Utility determines the required load reduction $D$ kWh that is to be delivered by the aggregator based on the system requirements.
The aggregator recruits consumers to deliver the required reduction. The mechanism that we proposed for the aggregator requires the consumers to report both their marginal utility and their baseline consumption. The aggregator uses the marginal utility reports to select consumers such that its overall cost is minimized while the load reduction target $D$ is met. A drawback of this mechanism in terms of implementation is that the consumers need not have knowledge of their true marginal utilities.

The new approach proposed here also avoids baseline estimation by requiring the consumers to report their baseline consumption, but the individual marginal utility need not be disclosed. The mechanism is minimal in terms of the information it elicits from the consumers because it does not require either historical data or any additional infrastructure. The authors in \cite{vuelvas2018limiting} use a similar problem formulation to ours and propose an incentive based DR mechanism, but do not address how the reward price is set, how the consumers are selected so that the DR service is delivered reliably and the cost aspect of the mechanism. In addition they also ignore any randomness in the consumption of the consumers. Here, we consider all of the above aspects and the randomness in consumption. We also provide comparison with other baseline estimation methods.

While some parallels can be drawn with dynamic pricing based DR mechanisms \cite{jacquot2018,muratori2016,yoon2014}, the setting we consider here is different. These mechanisms essentially influence consumers by using time varying prices to alter their energy consumption so that system objectives are met. On the contrary, the central problem we consider is to recruit DR resources that can deliver a certain amount of load reduction at certain times of a month which coincides with peak load conditions. This requires the estimation of consumer baseline because measuring load reduction requires a baseline. Hence baseline estimation becomes a primary concern in our setting whereas such a requirement does not arise in the dynamic pricing DR setting.

The remainder of this paper is organized as follows. In Section \ref{sec:setup}, we introduce the consumer model and the incentive-based self-reported DR program. In Section \ref{sec:opt-for}, we solve for the optimal consumer forecast and characterize baseline inflation for a quadratic utility function and a quadratic penalty with and without deadband. In section \ref{sec:socost} we discuss SO's cost. In Section \ref{sec:CAISO-comp}, we compare self-reported baseline with other conventional baseline estimation methods.
Finally, we conclude in Section~\ref{sec:conclusion}.

\section{Problem Setup}\label{sec:setup}

In this section, we describe the market model, the consumer model and the
incentive-based DR program. A summary of the notations is given in
Table~\ref{tab:notation}.

\begin{table}\caption{Notation}
\begin{center}\small
\begin{tabular}{cl}
\toprule
$q$ & Energy consumption of consumer\\
$\theta$ & Exogenous random variable\\
$u$ & Utility of consumer expressed in monetary units\\
$\pi_0$ & Retail price of energy \\
$\pi_2$ & Reward/kWh awarded to consumer $k$\\
$f$ & Baseline report of consumer\\
$\pi^*$ & Threshold Market Clearing Price (TMC)\\
$p$ & Probability of consumer being signaled \\
$R$ & Reward function for load reduction \\
$\Phi$ & Penalty function for deviation from baseline \\
$\Pi$ & Inverse supply function\\
$Q_0$ &  Peak load  \\
\bottomrule
\end{tabular}
\end{center}
\label{tab:notation}
\end{table}

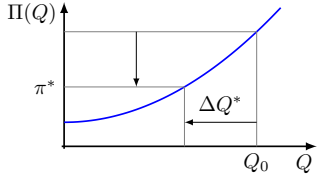
\begin{figure}
\centering
\scalebox{0.8}{
\begin{tikzpicture}[scale=4]
      \draw[->,>=latex, thick] (-0.01,0) -- (1.05,0) node[below,pos=.95] {$Q$};
      \draw[->, >=latex, thick] (0,-0.01) -- (0,0.6);
      \draw[domain=0:0.9,smooth,variable=\x,blue,thick] plot ({\x},{0.1+\x*\x/1.7});
      \draw[help lines] (0,{0.1+0.5*0.5/1.7}) -- (0.5,{0.1+0.5*0.5/1.7});
      \draw[help lines] (0,{0.1+0.8*0.8/1.7}) -- (0.8,{0.1+0.8*0.8/1.7});
      \draw[help lines] (0.5,0) -- (0.5 ,{0.1+0.5*0.5/1.7});
      \draw[help lines] (0.8,0) -- (0.8 ,{0.1+0.8*0.8/1.7});
      \draw[<-,>=latex] (0.5,0.1) -- node[above] {$\Delta Q^*$} (0.8,0.1);
      \draw[<-,>=latex] (0.3,{0.1+0.5*0.5/1.7}) -- (0.3,{0.1+0.8*0.8/1.7});
      \path (0.8,0) node[below] {$Q_0$};
      \path (0,0.55) node[left] {$\Pi(Q)$};
      \path (0,{0.1+0.5*0.5/1.7}) node[left] {$\pi^*$};
\end{tikzpicture}
}
\caption{Inverse Supply Curve and Threshold Market Clearing (TMC) Price $\pi^*$}
\label{fig:supcurv}
\end{figure}

\subsection{Market Model}
\label{sec:mar-mod}

The market model is represented by the wholesale market's
inverse supply function $\Pi(Q)$ which provides the energy price
as a function of the net energy transacted in the wholesale market,
see Fig.~\ref{fig:supcurv}.
The market inverse supply function is assumed to be convex with respect to $Q$ and monotone increasing, \emph{i.e.} with positive derivative $\Pi'(Q) > 0$. The \emph{threshold market clearing price} $\pi^*$ (TMC price), as defined earlier, is the market price above which it is profitable for the SO to call the consumers. Given the inverse supply curve of the market, this price can be computed a priori. In scenarios where the inverse supply function is not available a priori, the SO can estimate it using data from past twelve months. This is typical of many system operators such as the CAISO which publishes threshold market clearing prices for the next target month using data from past twelve months. The assumption we make is that this estimate is reflective of the true TMC price.

Note that the calculation of the TMC price from the inverse supply function may not be straightforward for a network model. This would require a detailed analysis of how the congestion constraints influence the Locational Marginal Prices (LMPs) of the nodes and is model specific. The main results of this paper will still hold provided $\pi^*$, \emph{i.e.} the threshold market clearing price for a node in the network, is determined via the network model. It is beyond the scope of this paper to discuss in detail the determination of this price under such more complicated settings. The main goal of this paper is to illustrate the idea of self-reported baseline and its effectiveness.

\subsection{Consumer Model}
\label{sec:con-mod}

Consider a residential consumer whose consumption is denoted by $q$. Let
$\theta$ be a random variable that is drawn from a continuous distribution. The utility of consumption of a consumer depends on this random variable. We
assume that the distribution of $\theta$ includes every possible source of
uncertainty. For example, $\theta$ could represent the consumer's state where
the consumer could either be at home or not. It could also model the
randomness induced due to external weather conditions like temperature. Let the
private utility function which is expressed in monetary units be $u(q,\theta)$, which is assumed to be a strictly concave monotone increasing in $q$. We also assume that the random variable $\theta$ is realized at the time when consumption is accomplished. Define the marginal utility $\mu(q,\theta)$ as follows:
\begin{equation}
\mu(q,\theta) =  \frac{\partial u(q, \theta)}{\partial q}.
\label{eq:margutil}
\end{equation}
Note that since $u(q, \theta)$ is monotone increasing and strictly concave in $q$,
we have:
\begin{equation*}
\forall q: \quad
\mu(q, \theta) > 0, \quad
\frac{\partial \mu(q, \theta)}{\partial q} < 0.
\end{equation*}

\subsection{Incentive-Based Demand Response Program}
\label{sec:mech}

The SO signals a DR event when the market price exceeds the TMC price. The novel DR program that we propose comprises a {\it self-reported baseline mechanism}. The mechanism has two stages which are as follows.

\paragraph*{{\bf Stage 1} (Reporting)} In this stage, the consumer self reports its baseline $f$ and the SO announces the following quantities:
\begin{enumerate}
\item the probability $p$ of calling a consumer,
\item the reward function $R(\pi_2,f,q)$ for reducing consumption $(f-q)$,
\item reward per unit reduction $\pi_2$ which is equal to the TMC price $\pi^*$,
\item the penalty function $\Phi(f,q)$ for consumers who deviate
from their reported baseline when they are not called for DR service.
\end{enumerate}
This penalty function $\Phi$ is critical to ensure that the consumers do not inflate their
baseline report. At the same time, the penalty should not discourage
participation by preventing lack of profitability for the participants.
Based on the reward per unit reduction $\pi_2$ and the penalty $\Phi(f,q)$,
each consumer submits the baseline report $f$.

\paragraph*{{\bf Stage 2} (DR Event)} In the second stage, a DR event is
triggered when the SO expects the market price to shoot above the TMC price. The SO then selects
randomly from the pool of recruited consumers and the selected consumers are signaled to reduce consumption.
The SO observes the aggregated consumption $Q$ of those selected consumers.
By the mechanism, the consumers who are signaled and reduce consumption
are paid $\pi_2$ per unit of reduction. However, those recruited consumers
that are not signaled are penalized for deviating from their reported
baseline as prescribed by the penalty function.

\subsubsection{Consumer Recruitment and Selection}

The objective of SO is to minimize its cost during DR events.
During a DR event, the load is at its peak $Q_0$, and is desirable to achieve
a load reduction of $\Delta Q^*$, which is the optimal load reduction
(Refer~Fig.~\ref{fig:supcurv}). The SO recruits $n$ sets of consumers. The recruitment is such that each set of consumers reduces load by $\Delta Q^*$ for the specified reward/kWh $\pi^*$. The consumers are tested before they are recruited. Here, the assumption is that the aggregate load reduction can be more reliably established than the individual load reduction which requires a reliable baseline estimate. Since the probability of selection or calling of each individual consumer is restricted to probability $p$, the number $n$ of such sets of consumers recruited satisfies $np =1$. When a DR event occurs, one set is randomly chosen and its members are signaled to reduce consumption. This recruitment and selection process ensures that one set is always chosen. Hence, the required level of reduction $\Delta Q^*$ is delivered during all DR events while satisfying the calling probability of each recruited consumer.


It is inconceivable that each set of consumers will exactly deliver $\Delta Q^*$ amount of reduction at the prescribed reward. Hence, in the proposed mechanism, the SO is allowed to adjust the selected consumers within the DR event window. If the price remains higher than the TMC price within the DR event then the SO can call more consumers till the price falls to the desired level. Note that this does require the SO to recruit some set of consumers who can respond on short notice. Such type of consumers can be recruited under the \emph{flexible resource} category.

Here, we provide a very simple example to illustrate how the consumers are grouped and selected. Consider the case where the consumers are identical and have a capacity to deliver
$0.5$ kWh of reduction when paid at $\pi^* = \$ 0.05/\text{kWh}$. Let the probability of calling a consumer be $p = 0.1$ and the optimal load reduction $\Delta Q^* = 10 \ \text{kWh}$. Then the SO would recruit $n = 1/p = 10$ groups each with a capacity to deliver $10 \ \text{kWh}$ of reduction when paid at $\pi^* = \$ 0.05/\text{kWh}$. This implies that each of these groups would contain $20$ such consumers and each of these groups will be called or selected by the probability $p = 0.1$ when a DR event occurs.

\begin{remark}
As stated earlier, for an incentive-based demand response program, determining the right baseline is very important as baseline can not be measured. In our mechanism, consumers self-report their baseline. No other information from consumers is needed other than baseline report. Hence, our mechanism is minimal in the information it elicits from the consumers.
\end{remark}

\subsubsection{Reward and Penalty Function}

The reward function in the mechanism is set as
\begin{equation}
R(\pi_2,f,q) =
\left\lbrace\begin{array}{ll}
\pi_2(f - q), & \text{if consumer is called,} \\
0, & \text{otherwise.}
\end{array} \right.
\label{eq:rew-fn}
\end{equation}

Thus, the SO pays the consumers according to the measured reduction
$f-q$, where $f$ is the consumer's baseline report and $q$ is the measured
consumption during the DR event. The reward per unit reduction is $\pi_2$. The consumer's penalty function is specified as follows,
\begin{equation}
\Phi(f,q) = \left\lbrace
\begin{array}{ll}
	0, & \text{if consumer is called,} \\
	\phi(f-q), & \text{otherwise.}
\end{array} \right.
\label{eq:pen-fn}
\end{equation}
where the penalty function $\phi$ in \eqref{eq:pen-fn} is chosen to be convex,
symmetric, and nonnegative with minimum value zero at the origin,
\emph{i.e.} it satisfies the following conditions:
\begin{equation}\label{eq:phicon}
\phi(0) = \phi'(0) = 0, \ \forall x: \phi(x)=\phi(-x), \
\phi''(x) > 0,
\end{equation}
where $\phi'$ and $\phi''$ denote the first and second derivative of the
penalty function $\phi$.

\subsection{Consumer's Optimization Problem}

The minimum expected cost incurred by a consumer is a function of
the baseline report $f$ and is given by
\begin{equation}\label{eq:H}
H(f) = {\mathds{E}_{\theta}}
\left[
\min_{q}
\left\{\pi_0q - u(q, \theta)  + \Phi(f,q) - R(\pi_2,f,q)\right\}
\right].
\end{equation}

The \emph{consumer's problem} is formulated as follows:
\begin{equation}
\label{opt-CP-1}
\text{\bf CP:}~~
\min_f ~ H(f).
\end{equation}

Hence, the consumer's problem is a two stage stochastic decision problem.
In the first stage the consumer decides the optimal baseline
report $f$, and in the second stage decides the optimal consumption $q$.
\begin{definition}
Let $f^*$ be defined as the baseline report that minimizes the cost that is incurred by the consumer, \emph{i.e.} $f^* = \arg \min~ H(f)$.
\label{eq:opt-f}
\end{definition}

\section{Optimal baseline report and inflation}
\label{sec:opt-for}

In this section, we derive an optimality condition for the consumer baseline
report that minimizes the expected cost of the consumer. The optimality condition has a
nice economic interpretation because it establishes that the baseline
report that minimizes the expected cost is such that the marginal utility
of the consumers equals the retail price of the electricity.

The consumer's optimization problem \textbf{CP} given by (\ref{opt-CP-1})
is a two-step stochastic decision problem. We characterize
consumers's consumption decisions corresponding to the second stage problem and
then obtain the optimality condition for the consumer's baseline
report by solving the first stage problem.

\subsection{The Consumer's Second Stage Problem}
\label{sec:secstage-decision}

The consumer has several choices. It can decide to participate or not to
participate in the DR program. If it decides to participate, then it can be
signaled to reduce its consumption or not signaled. This gives rise to three
possible scenarios for the second stage: {a)} consumer is not participating in
the DR program, {b)} consumer is participating in the program but is not
signaled to reduce consumption, {c)} consumer is participating in the program
and is signaled to reduce consumption. We obtain the optimal consumption for
each of the three cases assuming that the baseline report $f$ is given.  The
consumption when the consumer is not participating corresponds to the true
baseline. Hence, we use this value as the baseline to characterize inflation in
the DR program.

\paragraph{Consumer is not participating in the DR program}
In this case, $R = 0$ and $\Phi = 0$. Let $J^a(q,\theta)$ denote the
realized cost function for this case. It is then given by
\begin{equation} \label{eq:Ja}
J^a(q,\theta) = \pi_0 q - u(q,\theta),
\end{equation}
where $\pi_0$ is the retail price of electricity.
The optimal consumption is given by
\[ q^a(\theta) = \arg\min_q J^a(q, \theta), \]
which is a function of $\theta$ because its value is realized when
the consumption decision is made. Note that $q^a(\theta)$ is the solution of
the first order optimality condition,
\begin{equation}
\pi_0 - \frac{\partial u(q, \theta)}{\partial q} = 0.
\label{eq:opt-con-np}
\end{equation}
Hence, $q^a(\theta)$ is given by
\begin{equation}\label{eq:opt-1}
q^a(\theta) = \mu^{-1}(\pi_0,\theta),
\end{equation}
where $\mu^{-1}$ denotes the inverse function of the marginal utility,
see \eqref{eq:margutil}, that always exists for every $\theta$. Moreover,
since the consumer's utility is monotone increasing and concave in
$q$, the consumption $q^a(\theta)$ is always nonnegative.

\paragraph{Consumer is participating in the program but is not signaled
to reduce consumption}
The reward and penalty functions are given by \eqref{eq:rew-fn}
and \eqref{eq:pen-fn}. Let $J^b(f,q)$ denote the realized cost function
which is given by
\begin{equation}\label{eq:Jb}
J^b(f, q,\theta) = \pi_0 q - u(q,\theta) + \phi(f-q).
\end{equation}
As before, the value of $\theta$ is realized when the consumption decision is
made. In this scenario the realized cost is also a function of
the baseline report $f$ in addition to  the consumption decision and the
value of $\theta$.  The optimal consumption is given by
\[q^b(f, \theta) = \arg\min_q J^b(f, q, \theta), \]
and so it satisfies the first order optimality condition,
\begin{equation}
 \pi_0 - \frac{\partial u(q, \theta)}{\partial q} - \phi'(f-q) = 0.
\label{eq:opt-con-p-ndr}
\end{equation}
Hence, the optimal consumption satisfies the following implicit equation,
\begin{equation}\label{eq:qb}
q^b(f,\theta) = \mu^{-1}(\pi_0 - \phi'(f - q^b(f,\theta)),\theta),
\end{equation}
and $q^b(f,\theta)$ is also a function of $f$ because the deviation from $f$
incurs a penalty.

\paragraph{Consumer is participating in the program and is signaled to
reduce consumption}
Again, the reward and penalty functions are given by equations
\eqref{eq:rew-fn} and \eqref{eq:pen-fn}, respectively.
Let $J^c(q,\theta)$ denote the realized cost function which is given by
\begin{equation}\label{eq:Jc}
J^c(q,f,\theta) = \pi_0 q - u(q,\theta) - \pi_2(f-q).
\end{equation}
The optimal consumption is given by
\[q^c(f, \theta) = \arg\min_q J^c(f, q, \theta).\]
So $q^c(f, \theta)$ is the solution of
\begin{equation}
\pi_0 - \frac{\partial u(q, \theta)}{\partial q} + \pi_2 = 0.
\label{eq:opt-con-p-dr}
\end{equation}
Hence, the optimal consumption $q^c$ is independent of $f$ and is given by
\begin{equation}
q^c(\theta) = \mu^{-1}(\pi_0 + \pi_2,\theta).
\label{eq:opt-con-p-dr-exp}
\end{equation}

The relation between the consumptions for the three different cases
$q^a(\theta)$, $q^b(\theta,f)$, $q^c(\theta)$ and the consumer's baseline
report $f$ are stated in the following lemma.

\begin{lemma} \label{lem:qa-qb-f}
The optimal consumptions for the three cases $q^a(\theta)$, $q^b(\theta,f)$,
$q^c(\theta)$ satisfy the conditions (i) $q^c < q^a$ and (ii) $q^a < q^b < f$ or $f < q^b < q^a$ for every $\theta$.
\end{lemma}
\begin{proof}
Refer Appendix.
\end{proof}

As a result of Lemma~\ref{lem:qa-qb-f}, a rational consumer that is participating in the DR program and is signaled always provides a load reduction with respect to its true baseline consumption $q^a$.
However, according to this lemma, a consumer that is participating and not signaled for reduction may inflate its consumption near to its inflated baseline report to avoid the penalty and gain from the inflated baseline when called for reduction. This behaviour needs to be controlled.

\subsection{The Optimal Baseline Report}

Let $p$ denote the probability that the consumer is signaled to reduce
when a DR event occurs.  The expected cost that is incurred by the
consumer (\ref{eq:H}) can be expressed in terms of the probability $p$ as
follows:
\begin{equation}
H(f) =
p\mathds{E}_{\theta} J^c(f, q^c, \theta)  +
(1-p)\mathds{E}_\theta J^b(f, q^b, \theta),
\label{eq:exp-cost}
\end{equation}
and it follows that the optimal baseline report $f^*$ minimizes this $H(f)$. In the following lemma, we show that the consumer's expected cost $H(f)$ is a
convex function of its argument $f$.

\begin{lemma} \label{lem:conv-cost}
The consumer's expected cost $H(f)$ is a (strictly) convex function of its
argument $f$ if and only if the penalty $\phi$ is (strictly) convex.
\end{lemma}
\begin{proof}
Refer Appendix.
\end{proof}

Since the penalty function $\phi$ was chosen to be convex, the consumer's expected cost $H(f)$ is also convex.
\begin{definition}[Consumer's Expected Marginal Utility]
The consumer's expected marginal utility under the incentive-based
self-reported DR program is given by
\begin{equation}
M(f) = p\mathds{E}_{\theta}  \frac{\partial u(q^c,\theta)}{\partial q} + (1 - p) \mathds{E}_{\theta} \frac{\partial u(q^b,\theta)}{\partial q}.
\end{equation}
\end{definition}

The consumer's expected marginal utility is a function of the baseline report $f$, because the consumption $q$ is a function of $f$. For example, if the consumer is participating in the DR program and is signaled, then its consumption is $q^b(f,\theta)$ which solves the implicit equation (\ref{eq:qb})
and does depend on $f$. The following theorem establishes the optimality condition for the optimal baseline report $f^*$ in terms of $M(f)$,
\begin{theorem} \label{lem:f-opt-cond}
The optimal baseline report $f^*$ satisfies $\pi_0 = M(f^*)$ and
is a global minimizer of the cost function $H(f)$. Moreover,
the minimizer is unique when $\phi$ is strictly convex.
\end{theorem}
\begin{proof}
Refer Appendix.
\end{proof}

The optimality condition obtained in Theorem~\ref{lem:f-opt-cond} has a nice interpretation from the classical consumer theory in economics~\cite{Krugman2012}. The optimal baseline report $f^*$ is such that the consumer's marginal utility equals the retail price of the electricity. Given $f^*$, the expected reward per unit of energy (in kWh) paid for the expected load reduction provided by a consumer is given by

\begin{equation}
\frac{\pi^*(f - \mathds E_{\theta} q^c(\theta))}
    {\mathds E_{\theta} q^a(\theta) - \mathds E_{\theta} q^c(\theta)}
= \pi^* + \pi^* \frac{\mathds E_{\theta}\delta f(\theta)}
                       {\mathds E_{\theta} q^a(\theta) -
    \mathds E_{\theta} q^c(\theta)},
\end{equation}
where $\delta f(\theta) = f-q^a(\theta)$ is defined to be the inflation of the baseline report. From the second term, we infer that the expected inflation of the baseline report should be small to avoid a large excess payment.

\subsection{Control of Baseline Inflation}

Here, we show that the penalty function in combination with randomized calling allows the SO to control the inflation of the optimal baseline report $\delta f^*(\theta) = f^*-q^a(\theta)$. First we establish that penalty is necessary and then show that with a penalty, the probability of calling $p$ provides us a lever to control baseline inflation.

\subsubsection{Optimal Baseline Report without Penalty}
In this case, the optimality condition for the optimal baseline report $f^*$ is given by
\begin{equation}
\frac{d H(f)}{d f} =
p \mathds{E}_{\theta}  \frac{d J^c(f, q^c, \theta)}{d f} = -p \pi_2.
\end{equation}
Since the sensitivity of the consumer's cost $H(f)$ is negative with respect to $f$,
it indicates that the consumer will report a very high baseline.

\subsubsection{Optimal Baseline Report with Penalty}

The introduction of a penalty function allows us to control the inflation in the
baseline report by adjusting the probability of calling. This result is shown in the following
lemma.

\begin{theorem} \label{lem:f-opt-novar}
Let the penalty function $\phi$ be a quadratic function such that $\forall x: \phi''(x) = 1/\lambda$. Then the measurable inflation in the optimal baseline report $\delta \tilde{f}^*(\theta) = f^* - q^b(f^*,\theta)$ satisfies
\begin{equation*}
\lim_{p\rightarrow 0} \mathbb E_{\theta} \delta \tilde{f}^*(\theta) = 0.
\end{equation*}
And when  $\frac{\partial u^2(q,\theta)}{\partial q^2} = -1/d$,
\begin{equation*}
\lim_{p\rightarrow 0} \mathbb E_{\theta} \delta f^*(\theta) = 0.
\end{equation*}
\end{theorem}
\begin{proof}
Refer Appendix.
\end{proof}

For specific consumer utility and penalty functions, an explicit expression for
the expected baseline report inflation can be obtained. The following theorem
provides this expression for the case where the consumer's utility and the
penalty function are both quadratic.

\begin{theorem} \label{th:f-opt-ub}
Let the consumer's utility $u$ and the penalty function $\phi$ be quadratic
functions such that
\begin{enumerate}
\item[i)] $\forall (q,\theta): \frac{\partial u^2(q,\theta)}{\partial q^2} = -1/d$,
\item[ii)] $\forall x: \phi''(x) = 1/\lambda$,
\end{enumerate}
where $d$ and $\lambda$ are positive scalars,
then the expected inflation of the baseline report is given by
\begin{equation}
\mathds E_{\theta} \delta f^*(p) =
f^* - \mathds E q^a(\theta) =
\left(d+\lambda\right)\frac{p\pi_2}{1-p}.
\end{equation}
\end{theorem}
\begin{proof}
Refer Appendix.
\end{proof}

The law of diminishing marginal utility establishes that the marginal
utility declines with increase in consumption~\cite{Krugman2012}.
In Theorem~\ref{th:f-opt-ub}, $1/d$ is the rate of diminishment of the
consumer's marginal utility, and it is a private feature of the consumer
that cannot be modified by the system operator.
Unlike $d$, $\lambda$ is a parameter of the DR program, because
it defines the quadratic penalty function, \emph{i.e.}
$\phi(x)=x^2/(2\lambda)$.
Hence, the SO can choose $\lambda$ in the  design of the DR program.
Since $\lambda > 0$, a lower bound for the expected inflation of
baseline report is obtained by setting $\lambda=0$,
\begin{equation}
\mathds E_{\theta} \delta f^*
= f^* - \mathds{E}_{\theta} q^a(\theta)
\geq \frac{dp\pi_2}{1-p}.
\end{equation}
Moreover, by choosing the parameter of the penalty function $\lambda$ to be
small enough, the expected inflation of the baseline report can be made
arbitrarily close to its lower bound. Note that this lower bound is a function of
$p$ and is decreasing with $p$. Consequently, by choosing $\lambda$ and $p$ to
be small the baseline inflation can be controlled.

Here, we provide a simple numerical example to validate the above results. In this example, $u = cq - (0.5/d)q^2$, where $c = \$.5/\text{kWh}$ and $d \in \{0.1,0.2,0.3,0.4\}$ in $(\$/\text{kWh}^2)^{-1}$. The retail price $\pi_0 = \$ 0.12/\text{kWh} $ and the TMC price $\pi^* = \$ 0.05/\text{kWh}$ and are typical values (Refer \cite{eia-gov}). These set of parameter values correspond to a typical price sensitivity value of $\sim -0.3$ \cite{king2003predicting, reiss2005household}. The penalty coefficient $\lambda = 0.1 \ (\$/\text{kWh}^2)^{-1}$. The probability $p$ is chosen to be $p = 0.1$. Table \ref{table:dfcomp} summarizes the simulation results and how it compares with the theoretical results for this example.

\begin{table}[h]
\centering
\caption{Baseline Inflation, $\delta f^*$}
\begin{tabular}{ccccc}
\toprule
$d$ & 0.1 & 0.2 & 0.3 & 0.4 \\
\midrule
$\delta f^*$ (theory) &  $0.0011$ & $0.0017$ & $0.0022$ &  $0.0028$\\
\midrule
$\delta f^*$ (simul.) & $0.0012$ & $0.0017$ & $0.0023$ &  $0.0028$\\
\bottomrule
\end{tabular}
\label{table:dfcomp}
\end{table}

\subsection{Ensuring Individual Rationality with a Deadband}

\begin{figure}
\centering
\scalebox{0.8}{
\begin{tikzpicture}[scale=1]
      \draw[->,>=latex,thick] (-2,0) -- (2,0) node[below] {$q$};
      \draw[->,>=latex,thick] (-2,0) -- (-2,2) node[left] {$\Phi(q)$};
      \draw[domain=0:1,smooth,variable=\x,blue, thick] plot ({.5+\x},{2*\x*\x});
      \draw[domain=-1:0,smooth,variable=\x,blue, thick] plot ({-.5+\x},{2*\x*\x});
      \draw[blue, line width=1pt] (-0.5,0) -- (0.5,0);
      \draw (0,0) -- (0,2);
      \draw[help lines] (-0.5,0) -- (-0.5,1.5);
      \draw[help lines] (0.5,0) -- (0.5,1.5);
      \draw[<->,>=latex] (-0.5,1) -- (0.5,1) node[pos=.7,above] {$2\epsilon$};
      \draw (0,0) -- (0,-0.1) node [below] {$f$};;
\end{tikzpicture}
}
\caption{Penalty function with deadband}
\label{fig:penalty}
\end{figure}
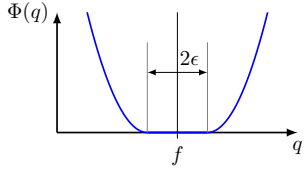

The DR program is not guaranteed to be individually rational from the
point of view of a single consumer because of the presence of uncertainty
$\theta$. When a consumer is not called, it consumes $q^b(f,\theta)$ which
varies with $\theta$ and is different from $f$, as it was shown in
Lemma~\ref{lem:qa-qb-f}.
As a result the consumer incurs a penalty and the mechanism is not
guaranteed to be individually rational. This is not an issue in the absence of
uncertainty. Individual rationality of the program can be ensured by
introducing a deadband in the penalty as illustrated in
Figure~\ref{fig:penalty}. A penalty with deadband can be expressed as,
\begin{equation}
\phi(f - q) =
\left\{
\begin{array}{ll}
\frac{(\vert f - q \vert - \epsilon)^2}{2\lambda}, &
\textrm{if}\ \vert f - q \vert \geq \epsilon, \\
0,  & \textrm{otherwise.}
\end{array} \right.
\label{eq:penalty-db}
\end{equation}

It is evident from such a design that there exists a deadband width $\epsilon$
such that the mechanism is individually rational. But a deadband worsens the
inflation in baseline. In the theorem below, we provide an upper bound for the
inflation in baseline when the penalty function has a deadband. The upper bound
explicitly proves that the baseline inflation can worsen with the introduction
of a deadband. But this trade off has to be made to guarantee individual
rationality.

\begin{theorem} \label{th:f-opt-ub-db}
Let the consumer's utility $u$ be a quadratic function such that
\begin{equation*}
\forall (q,\theta):
\frac{\partial u^2(q,\theta)}{\partial q^2} = -1/d.
\end{equation*}
The penalty function $\phi$ is defined in \eqref{eq:penalty-db},
where $d$ and $\lambda$
are positive scalars. Let $\max\{q_{\max} - \mathds{E}_\theta q^a(\theta), \mathds{E}_\theta q^a(\theta) - q_{\min} \}\leq \epsilon$, where $q^a(\theta) \in [q_{min}, q_{max}]$, then the expected inflation of the
baseline report is bounded by
\begin{equation}
\mathds E_{\theta} \delta f^*(p) =
f^* - \mathds E q^a(\theta) \leq
\left(d+\lambda\right)\frac{p\pi_2}{(1-p)} + \epsilon,
\end{equation}
and the mechanism is individually rational.
\end{theorem}
\begin{proof}
Refer Appendix.
\end{proof}

\section{SO's cost}
\label{sec:socost}

The SO's overall cost includes four terms: the cost to purchase power
from the wholesale market, the payment for DR services, the retail energy payments and the recruitment cost. We ignore the recruitment cost for the initial analysis here. This allows us to mathematically derive an order approximate expression, with respect to $p$, for the resultant cost. Using this we show that $p$ can be used as a lever to control SO's cost as well. This allows the SO to achieve an almost optimal cost in this case by choosing a very small value for $p$.

We then discuss the case where the recruitment cost is non-trivial. Here, we show that $p$ is restricted as a lever for controlling SO's cost. This is because the recruitment costs becomes unbounded as $p \rightarrow 0$. However, we show, for a typical DR scenario, that the SO's cost is decreasing with $p$ up to a certain threshold value. This threshold value is small enough that the cost can be significantly reduced by choosing this threshold as the selection probability. This suggests that the SO can still reduce its cost to significantly lower levels for typical DR scenarios.

\subsection{Without Recruitment Cost} Let $Q$ denote the net energy purchased from wholesale market, $\Pi(Q)$ the wholesale market's price, $\Delta\tilde{Q}$ the measured net load reduction provided by the called DR resources, and $\pi_0$ the retail energy price. Then, the SO's overall cost, ignoring the recruitment cost, is given by
\begin{equation}
J_{SO} = \Pi(Q)Q + \Pi(Q) \Delta \tilde{Q} - \pi_0 Q.
\label{eq:supplycost}
\end{equation}

Let $Q_0$ denote the overall load had the DR resources not been called to
reduce load and $\Delta Q$ the true net reduction provided by the called DR
resources, then $Q_0 = Q + \Delta Q$ and the SO's cost can also be written as
follows,
\begin{align}
J_{SO} & = (\Pi(Q_0-\Delta Q)-\pi_0) (Q_0-\Delta Q) + \nonumber\\
&  \qquad \Pi(Q_0 - \Delta Q) (\Delta Q + \Delta\tilde{Q} - \Delta Q),
\label{eq:supplycost-1}
\end{align}
where $\Delta \tilde{Q} - \Delta Q$ corresponds to the inflation in net load
reduction estimate, which arises from the inflation in baseline estimates of
the recruited DR providers.

From Theorem~\ref{th:f-opt-ub}, it follows that the inflation $\Delta \tilde{Q} - \Delta Q$ is $O(p)$ where $p$ is the probability of calling a consumer, which is a design variable of the DR mechanism. Hence, in this case, $\min_{\Delta Q, p = 0} J_{SO} = J^*_{SO}$, i.e., the SO's optimal cost can be achieved by driving the probability to zero. And so the optimal reduction $\Delta Q^*= \arg \min_{\Delta Q, p = 0} J_{SO}$. From the convexity of $J_{SO}$, when $p$ is zero, it follows that $\Delta Q^*$ satisfies the first order condition,
\begin{equation}
\Delta Q^* =  Q_0 - \Pi'^{-1}({\pi_0}/Q_0).
\label{eq:DQ-Q}
\end{equation}
The market price corresponding to $Q_0 - \Delta Q^*$ is exactly the TMC price $\pi^*$ because $Q_0 - \Delta Q^*$ is the optimal reduction. Consequently, $\Delta Q^*$ satisfies
\begin{equation}
\pi^* = \Pi(Q_0 - \Delta Q^*).
\label{eq:DQ-pi}
\end{equation}

The SO recruits $n=1/p$ sets of consumers such that each set can provide $\Delta Q^*$ of load reduction when called for a DR event. Hence, the cost for the SO \eqref{eq:supplycost-1} when a particular set is called during a DR event is given by
\begin{align}
J_{SO} & = (\Pi(Q_0-\Delta Q^*)-\pi_0)(Q_0-\Delta Q^*) + \nonumber\\
& \qquad  \Pi(Q_0 - \Delta Q^*)
         (\Delta Q^* + \Delta\tilde{Q} - \Delta Q^*).
\end{align}

Using definition of $J_{SO}^*$,
\begin{equation}
J_{SO}  = J_{SO}^* + \Pi(Q_0-\Delta Q^*)
                      (\Delta\tilde{Q} - \Delta Q^*).
\end{equation}

Substituting for baseline inflation from Theorem \ref{th:f-opt-ub},
\begin{equation}
 J_{SO} = J_{SO}^* + \Pi(Q_0-\Delta Q^*)
                      \left(\bar{N}\bar{d}\frac{p\pi_2}{1-p}\right),
\label{eq:supplycost-3}
\end{equation}
where $\bar{N}$ is the number of consumers in the set
and $\bar{d}$ is the average rate of diminishment of the marginal
utility across the recruited consumers in the set, which is an unknown.
The SO chooses the reward rate as $\pi_2 = \pi^*$, and therefore
\begin{equation}
J_{SO} = J_{SO}^* + (\pi^*)^2
        \left(\bar{N}\bar{d}\frac{p}{1-p}\right) = J_{SO}^* + O(p).
\label{eq:supplycost-4}
\end{equation}

Note that with the inclusion of deadband to ensure individual rationality,
the SO's cost becomes,
\begin{equation}
J_{SO} = J_{SO}^* + (\pi^*)^2\left(\bar{N}\bar{d}\frac{p}{1-p}\right) + \pi^*\bar{N}\epsilon = J_{SO}^* + O(p+\epsilon).
\label{eq:supplycost-5}
\end{equation}
Thus, for this case, the SO's cost is $O(p)$ and $O(\epsilon)$ optimal and the SO's cost $J_{SO}$ approaches $J_{SO}^*$ when both $p$ and $\epsilon$ approach zero. This result suggests that the SO can achieve an almost optimal cost in this case by choosing a very small value for $p$.

\subsection{With Recruitment Cost}
Denote the recruitment cost per customer by $\pi_{\text{rec}}$. Let $N_T$ be the total number of consumers recruited. Then $N_T$ is given by
\begin{equation*}
N_T = \sum_{i=1}^n\bar{N}_i,
\end{equation*}
where  $n = 1/p$ is the number of groups and $\bar{N}_i$ is the number of consumers in group $i$. Including the total recruitment cost, which scales with $N_T$, SO's overall cost is given by
\begin{equation}
J_{SO} = \Pi(Q)\cdot Q + \Pi(Q)\cdot \Delta \tilde{Q} - \pi_0\cdot Q + \pi_{\text{rec}} \cdot N_T.
\label{eq:supplycost-6}
\end{equation}
 Note that, in this case, the optimal cost for SO $J^*_{SO} \neq \min_{\Delta Q, p = 0} J_{SO}$. The reason is that the last term grows unboundedly as $p \rightarrow 0$. This also suggests that, in this case, an almost optimal cost cannot be achieved by choosing $p$ to be very small. This is illustrated in the example below.

We provide a simple example here to illustrate how the SO's cost varies with $p$ when the recruited consumers provide $\Delta Q^*$ reduction and when the recruitment cost is non-trivial. In the example we consider here, $c = 5\times 10^2 \ \$/\text{MWh}$, $\pi_0 = \$ 120/\text{MWh}$, $Q_0 = 8000$~MWh. We consider two different values for $d$, i.e., $ d = 0.1, d = 0.01$. The values of $d$ are derived from demand reduction provided by typical customers assuming the payment to be $\$ 100/\ \text{MWh}$. The two $d$ values correspond to large industrial customers and commercial places like retail stores etc. respectively \cite{kiliccote2008installation}. We assume that the supply ranges from $5000$~MWh to $8000$~MWh. Using the representative supply curve from \cite{hausman2006lmp} we approximate the inverse supply curve for this range by $\Pi(Q) = aQ + bQ^2$ where $a = -0.0415$ in $\$/\text{MWh}$ and $b = 8.3 \times 10^{-6}$ in $\$/\text{MWh}^2$. For this supply curve and $Q_0 = 8000$~MWh, $\pi^* \sim \$ 100/\text{MWh}$ and $\Delta Q^* \sim 1200$~{MWh}. The reward payment of $\$ 100/\text{MWh}$ and the aggregate load reduction of $1200$~{MWh} are typical of DR programs spanning the region covered by a SO \cite{interconnection2017demand}.

The first two plots of Figure~\ref{fig:socost} provides the variation of SO's cost with respect to $p$ when the recruitment cost is $\pi_{\text{rec}} = \$ 2/\text{Customer}$ for different values of $d$ and the bottom row plot of Figure~\ref{fig:socost} provides the variation of SO's cost when the recruitment cost $\pi_{\text{rec}} = \$10/\text{Customer}$. The former recruitment cost, i.e. $\pi_{rec} = 2$, is based on typical service costs charged per customer on a monthly basis to recover the metering implementation and maintenance cost \cite{doris2011government}. In our case, we consider the worst-case scenario where the SO bears this cost instead of passing it on to the DR participants. Note that the approximation of the SO's cost by ignoring recruitment cost, as in the previous section, is a reasonable approximation of the SO's cost up to a certain threshold probability. This threshold probability is as low as $0.1$ and $0.2$ for the cases $d = 0.1$ and $d = 0.01$ respectively. These $d$ values and other parameter values are typical values as stated before. Hence we expect that, in a typical scenario such as this, a SO can still reduce its cost significantly by setting the calling probability equal to this threshold value.

\begin{figure}
\begin{tabular}{ll}
\includegraphics[width=.49\columnwidth]{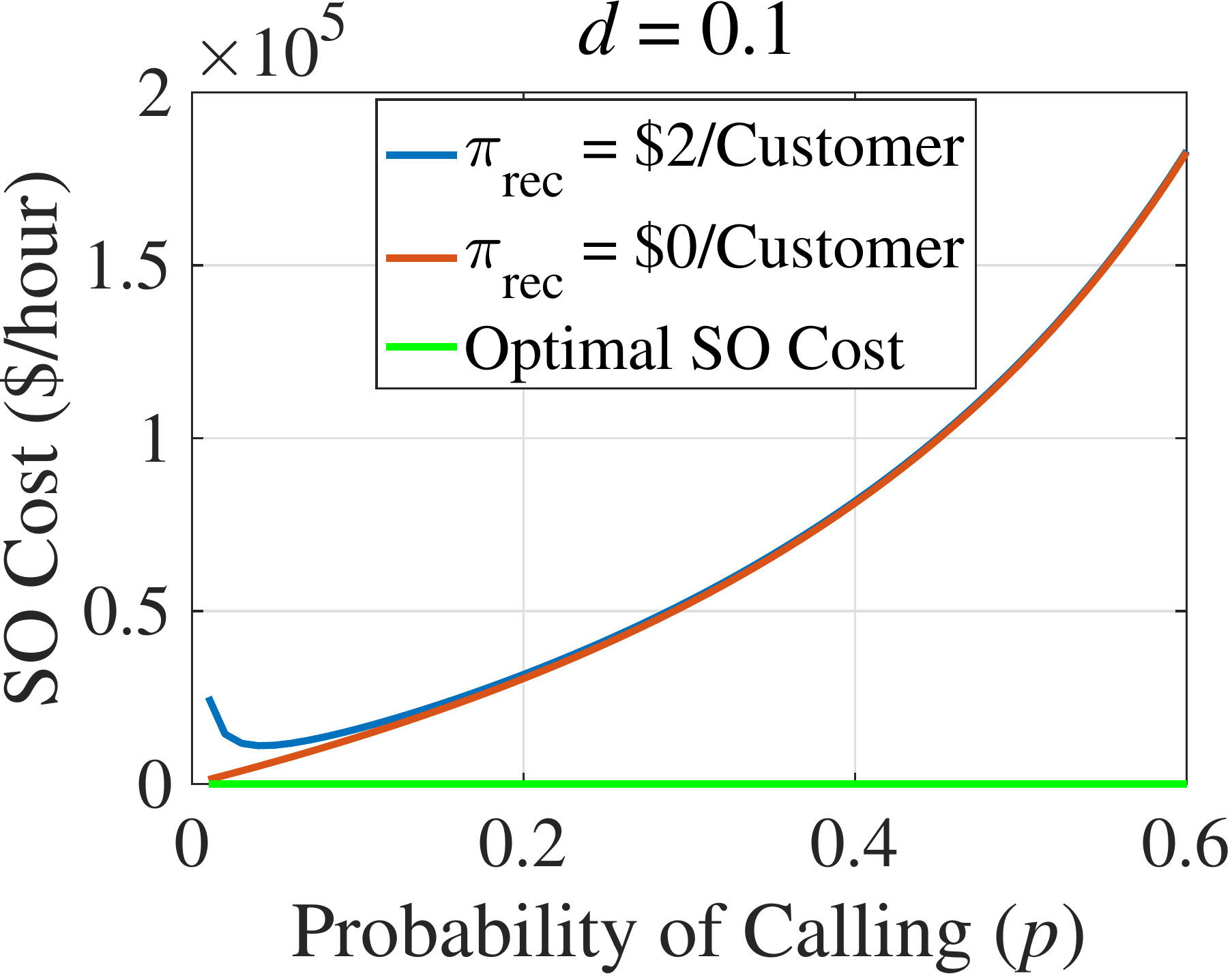} & \includegraphics[width=.49\columnwidth]{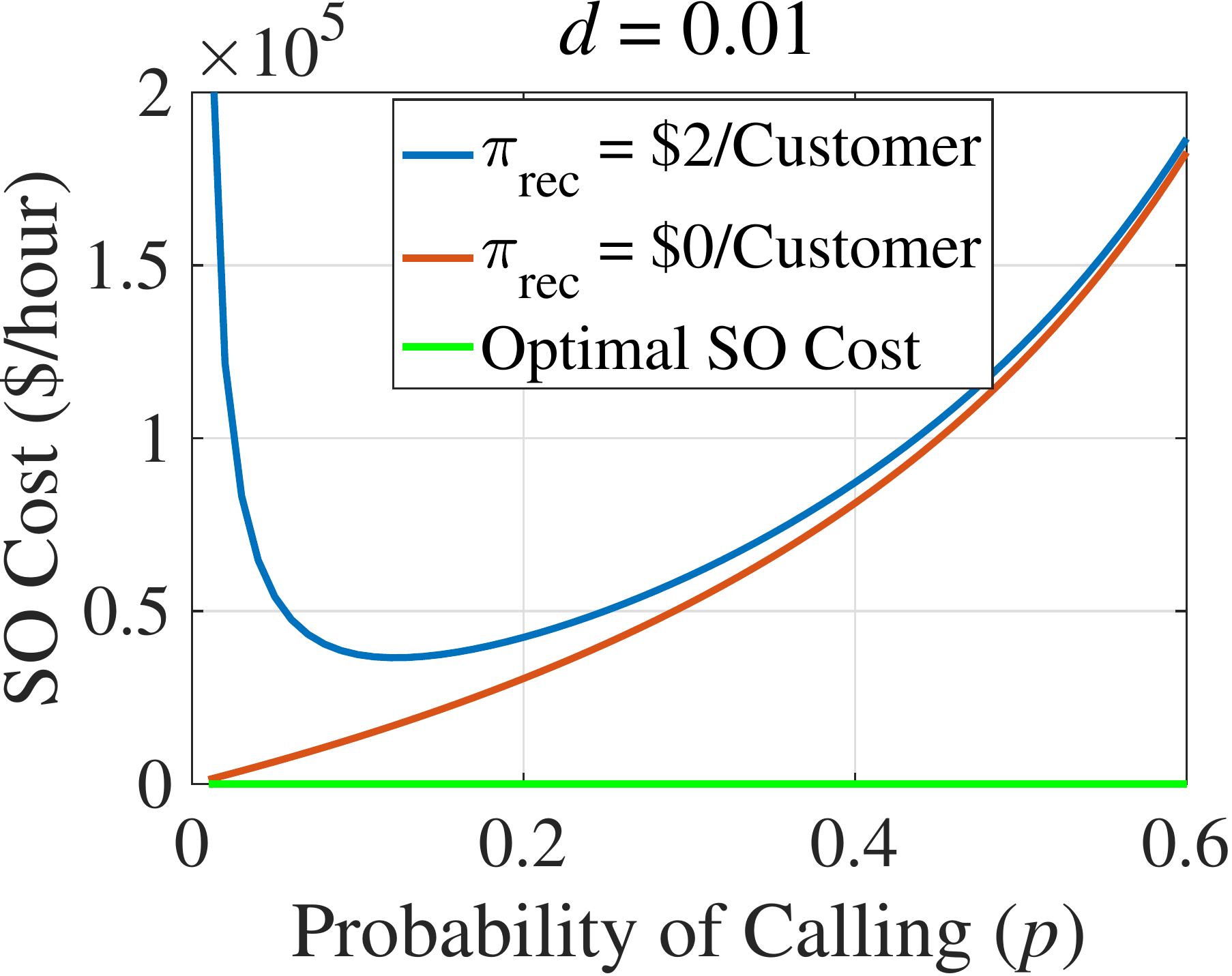} \\
 \includegraphics[width=.49\columnwidth]{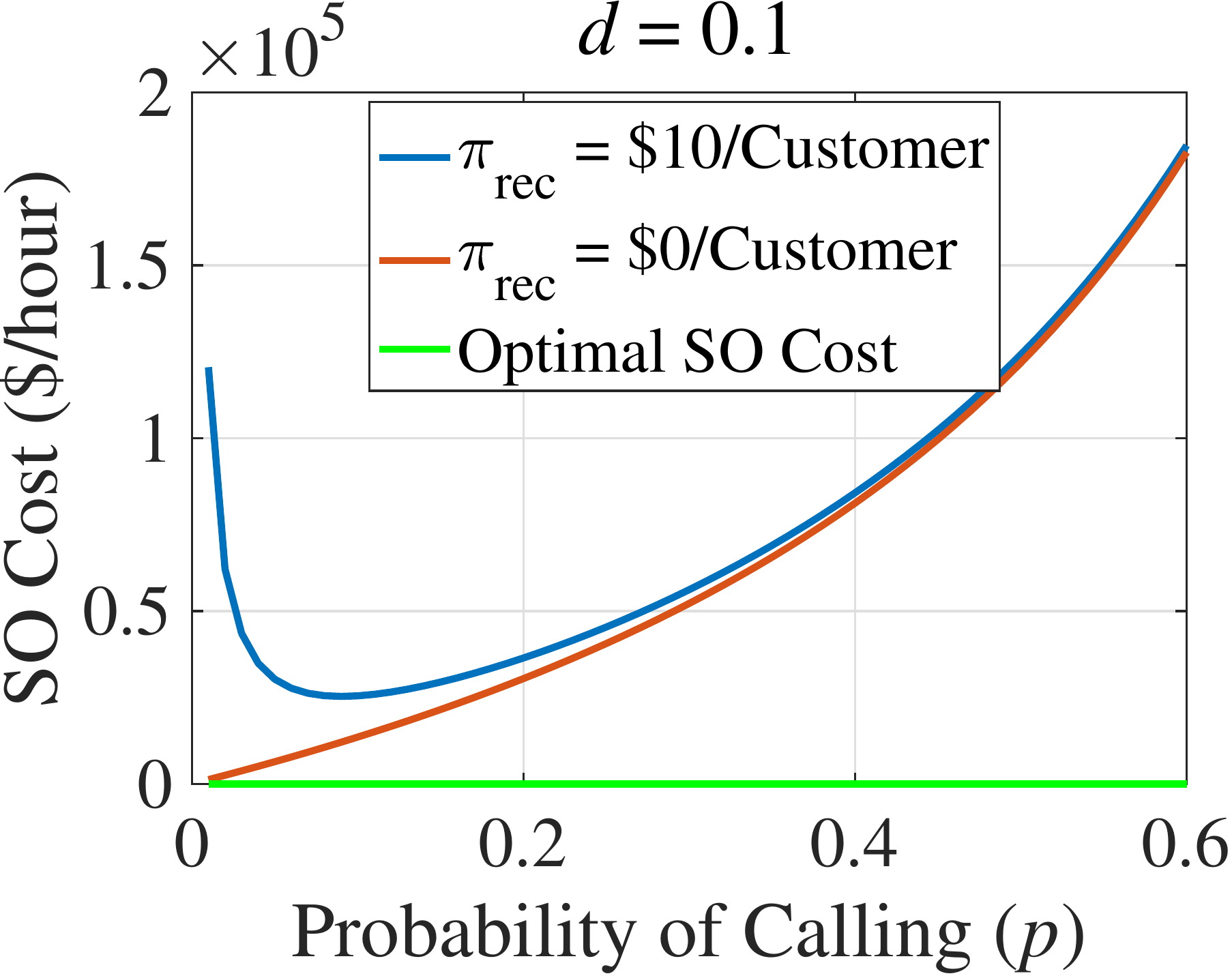} &  \includegraphics[width=.49\columnwidth]{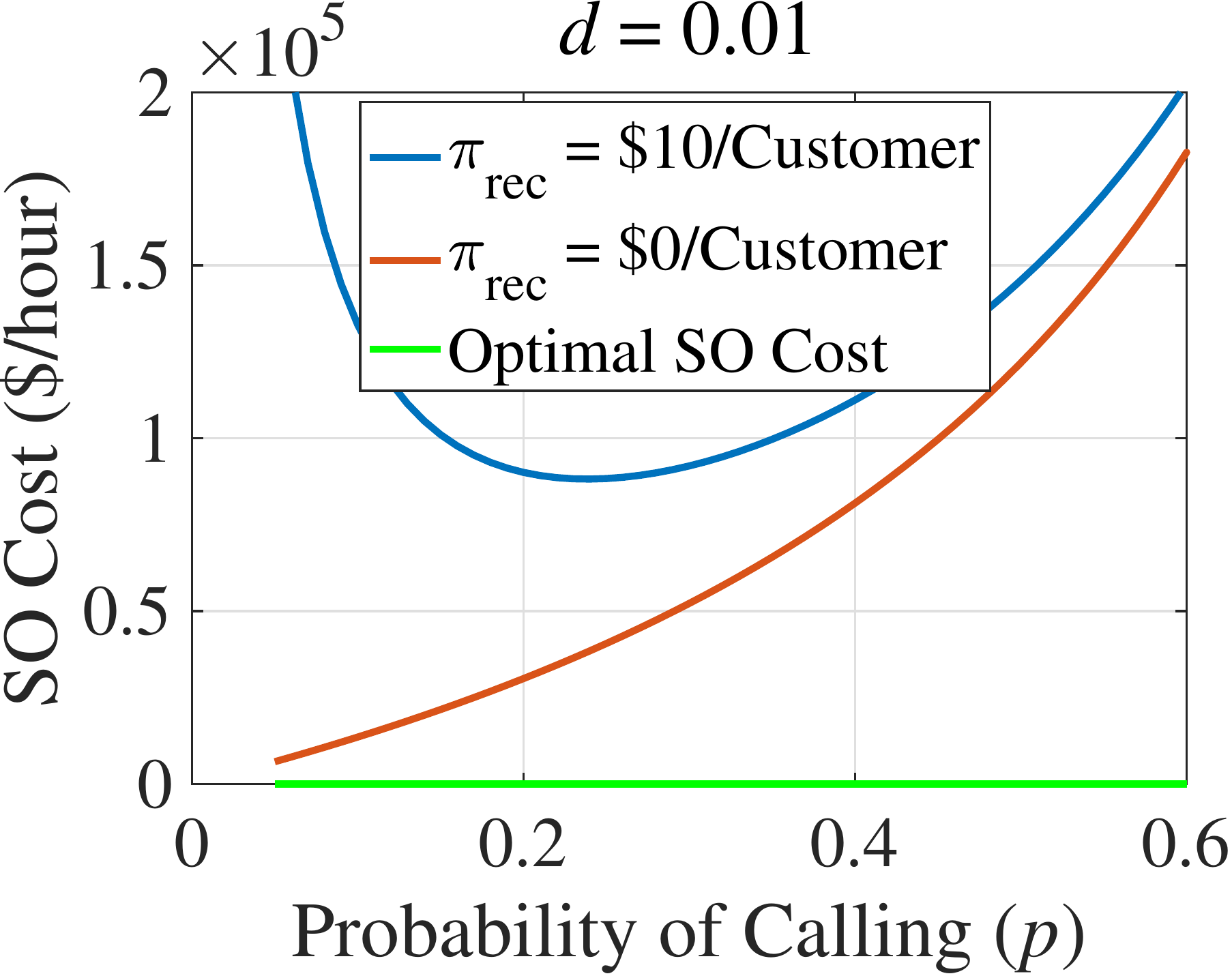}
\end{tabular}
\caption{SO cost vs $p$ when load reduction is $\Delta Q^*$ during a DR event. Top: $\pi_{\text{rec}} = \$ 2/\text{Customer}$, bottom:$\pi_{\text{rec}} = \$ 10/\text{Customer}$.}
\label{fig:socost}
\end{figure}

\section{Self-Reported Baseline vs. Other Baseline Estimation Methods}
\label{sec:CAISO-comp}

Here, we shall use the CAISO $m/m$ method~\cite{CaisoDR2017} for our comparative study. We emphasize here that a similar analysis applies to other estimation methods like control group methods. In CAISO's m/m method, the SO computes the average consumption of the most recent $m$ similar but non-event days and uses this average-based estimate as the baseline. Hence, the baseline estimate is a moving average of the consumption profile of the consumers. Typically this average-based estimate from past consumption data is corrected by an adjustment factor to account for any variation in the consumption pattern from the past. This adjustment factor is common to all baseline estimation methods and is highly recommended. As we shall see this factor is the primary cause for the existence of adverse incentives to inflate baseline. Hence, the analysis to follow equally applies to all current baseline methods that use an adjustment factor, which includes the control group methods.

As discussed before, the individual optimal consumption decision depends on whether the consumer is signaled or not for reduction on a particular day.
Also the baseline estimate used for the payments depends on the consumption in the days prior to the DR event. So the payments made during future DR events can influence the consumer to inflate their consumption during a non-event day. On a particular day, the consumer's benefit depends on whether the consumer is signaled or not. As explained in Section~\ref{sec:secstage-decision}, if the consumer is participating in the DR program and is signaled to reduce, its total cost is
$J^c(q,f,\theta)$ given by (\ref{eq:Jc}) where $f$ is the baseline estimate obtained by the CAISO $m/m$ method.  However, if the consumer is not signaled, its cost is $J^a(q,\theta)$ (see equation \ref{eq:Ja}). Unlike the incentive-based DR program with self-reported baseline, the CAISO
program does not impose a penalty, and the consumer's cost is the same as if it were not participating in the DR program, when it is not called.

Let $\mathcal T_N$ denote the set of most recent $m$ similar but non-event days, and let $f^c$ denote the baseline calculated in CAISO's $m/m$ model.  Then,
\begin{equation}
f^c = \frac{1}{m} \sum_{\tau\in\mathcal T_N}  q_{\tau},
\label{eq:caiso-baseline}
\end{equation}
where $\{q_{\tau}:\tau\in\mathcal T_N\}$ is the set of consumption
for the near past $m$ similar non-event days. The baseline estimation $f^c$ is multiplied by an adjustment factor $C_f$ to account for any variation in
the consumption pattern. Hence, the CAISO baseline with adjustment factor is given by
\[
\bar{f}^c = f^c C_f,
\]
where $f^c$ was defined in \eqref{eq:caiso-baseline}.
Let $q$ denote the consumption on the current DR event day, and $q^-$ the consumption on the day
before. Let $\mathcal T_E$ be the set of days before the days in the set $\mathcal T_N$,
and define $f^- = \frac{1}{m}\sum_{\tau\in\mathcal T_E} q^-_{\tau}$ as
the average consumption of the days in the set $\mathcal T_E$. The correction factor for the current DR day is then computed as
\begin{equation}
C_f = \frac{q^-}{f^-}.
\end{equation}

Typically, the consumers are signaled a day ahead of the DR event. So, the reward during the DR event on the current day can influence the consumer to inflate its day-ahead consumption $q^{-}$. The day-ahead consumption is obtained by minimizing the joint cost of the current DR day and the day before with respect to $q^{-}$, as these are the only two terms in the overall cost of the consumer that $q^{-}$ can influence. The joint cost for the two days is given by $J^a(q^-,\theta) + J^c(q,\theta)$, where
\begin{align*}
J^a(q^-,\theta) &= \pi_0 q^- - u(q^-,\theta),\\
J^c(q,\bar f^c,\theta) &= \pi_0 q - u(q,\theta) - \pi_2 (\bar f^c - q).
\end{align*}

Here, for illustration purposes, we have assumed identical utility functions and retail price for both the days. The analysis can be trivially extended to the general case where they are not identical. The term $\pi_2 (\bar f^c - q)$ is the payment that the consumer receives for
reducing consumption and the value of $\theta$ is realized when the consumption
decision is made.

By definition it follows that the optimal consumption $q^{-*}$ on the day before the current DR event day is given by
\begin{equation}
q^{-*}(\theta,\bar f^c) = \arg\min_{q^-} (J^a(q^-) + J^c(q,\bar f^c)).
\end{equation}
From the first order optimality condition it follows that $q^{-*}$ should satisfy
\begin{equation}
\pi_0 - \mu(q^{-},\theta) - \pi_2 \frac{\partial \bar{f}^c(q^-)}{\partial q^-} = 0.
\end{equation}
On the DR event day, $f^c$ and $f^-$ are constants. This implies,
\begin{equation}
\frac{\partial \bar{f}^c(q^-)}{\partial q^-} = \frac{f^c}{f^-}.
\end{equation}
Hence, the optimal consumption on the day before the current DR event day is
\begin{equation}
q^{-*}(\theta) = \mu^{-1} \left(
\pi_0 - \pi_2 \frac{f^c}{f^-},\theta
\right).
\end{equation}

Using this result, we provide a lower bound for the expected value of baseline
inflation in the CAISO m/m method with adjustment factor, when the utility
function is quadratic, in the following lemma.

\begin{lemma} \label{lem:caiso_adj}
Let the consumer's utility $u$ be a quadratic function such that
$\forall (q,\theta): \frac{\partial u^2(q,\theta)}{\partial q^2}
= -1/d$ where $d$ is a positive scalar, then the
baseline report $\bar f^c$ satisfies
\begin{equation}
\mathds E_{\theta} (\bar f^c - q_a(\theta) ) > d\pi_2.
\end{equation}
\end{lemma}

\begin{proof}
Refer Appendix.
\end{proof}

In the proposed DR mechanism, the expected baseline inflation with quadratic utility and penalty function was obtained in Theorem~\ref{th:f-opt-ub-db}. From the discussion in the previous two sections, it follows that $\lambda$ and $p$ can be used as levers to control baseline. Hence by choosing $\lambda$ and $p$ to be sufficiently small and provided $\epsilon$ is not comparable to $d\pi_2$, which is the case when $\pi_2 \sim O(\pi_0)$, we get that
\begin{equation}
\mathds E_{\theta}
( f^* - q^a(\theta) ) \ll d \pi_2 < 
     \mathds E_{\theta} (\bar f^c - q_a(\theta) ).
\end{equation}

Thus, in the self-reported approach, we can ensure that the inflation in baseline per consumer is
significantly smaller compared to conventional baseline estimation methods, such as CAISO's m/m method, that uses an adjustment factor.

\section{Conclusion} \label{sec:conclusion}
We proposed a mechanism for incentive-based DR programs where the only
information that is elicited from each consumer is a self-report of its
baseline consumption. The mechanism entails a calling probability for each
consumer and a penalty when the consumer is not called. The mechanism provides the required service reliably by selecting a certain set of consumers during every DR event. We showed that the probability of calling and the penalty can be used to control the baseline inflation. We also justified that the mechanism's cost can be significantly reduced by deploying DR resources. Finally, we showed that the self-reported baseline estimates a better baseline estimate than conventional methods such as the CAISO's m/m method.

\bibliographystyle{IEEEtran}
\bibliography{Refs-BL-TSG}

\begin{thebibliography}{10}
\providecommand{\url}[1]{#1}
\csname url@samestyle\endcsname
\providecommand{\newblock}{\relax}
\providecommand{\bibinfo}[2]{#2}
\providecommand{\BIBentrySTDinterwordspacing}{\spaceskip=0pt\relax}
\providecommand{\BIBentryALTinterwordstretchfactor}{4}
\providecommand{\BIBentryALTinterwordspacing}{\spaceskip=\fontdimen2\font plus
\BIBentryALTinterwordstretchfactor\fontdimen3\font minus
  \fontdimen4\font\relax}
\providecommand{\BIBforeignlanguage}[2]{{%
\expandafter\ifx\csname l@#1\endcsname\relax
\typeout{** WARNING: IEEEtran.bst: No hyphenation pattern has been}%
\typeout{** loaded for the language `#1'. Using the pattern for}%
\typeout{** the default language instead.}%
\else
\language=\csname l@#1\endcsname
\fi
#2}}
\providecommand{\BIBdecl}{\relax}
\BIBdecl

\bibitem{Albadi2008}
M.~H. Albadi and E.~El-Saadany, ``A summary of demand response in electricity
  markets,'' \emph{Electric power systems research}, vol.~78, no.~11, pp.
  1989--1996, 2008.

\bibitem{federal2011demand}
F.~E.~R. Commission, ``Demand response compensation in organized wholesale
  energy markets,'' \emph{Final Rule Report}, 2011.

\bibitem{Joskow2012}
\BIBentryALTinterwordspacing
C.~D.~W. Paul L.~Joskow, ``Dynamic pricing of electricity,'' \emph{The American
  Economic Review}, vol. 102, no.~3, pp. 381--385, 2012. [Online]. Available:
  \url{http://www.jstor.org/stable/23245561}
\BIBentrySTDinterwordspacing

\bibitem{chakraborty2017distributed}
P.~Chakraborty, E.~Baeyens, and P.~P. Khargonekar, ``Distributed control of
  flexible demand using proportional allocation mechanism in a smart grid: Game
  theoretic interaction and price of anarchy,'' \emph{Sustainable Energy, Grids
  and Networks}, vol.~12, pp. 30--39, 2017.

\bibitem{mathieu2013residential}
J.~L. Mathieu, T.~Haring, J.~O. Ledyard, and G.~Andersson, ``Residential demand
  response program design: Engineering and economic perspectives,'' in
  \emph{European Energy Market (EEM), 2013 10th International Conference on
  the}.\hskip 1em plus 0.5em minus 0.4em\relax IEEE, 2013, pp. 1--8.

\bibitem{borenstein2002dynamic}
S.~Borenstein, M.~Jaske, and A.~Ros, ``Dynamic pricing, advanced metering, and
  demand response in electricity markets,'' \emph{Journal of the American
  Chemical Society}, vol. 128, no.~12, pp. 4136--45, 2002.

\bibitem{faruqui2011dynamic}
A.~Faruqui and J.~Palmer, ``Dynamic pricing and its discontents,''
  \emph{Regulation}, vol.~34, no.~3, pp. 16 -- 22, 2011.

\bibitem{CaisoDR2017}
CAISO, \emph{Demand Response User Guide. Version 4.3}, California ISO, May
  2017.

\bibitem{chao2010price}
H.~Chao, ``Price-responsive demand management for a smart grid world,''
  \emph{The Electricity Journal}, vol.~23, no.~1, pp. 7--20, 2010.

\bibitem{wolak2007residential}
F.~A. Wolak, ``Residential customer response to real-time pricing: The anaheim
  critical peak pricing experiment,'' \emph{Center for the Study of Energy
  Markets}, 2007.

\bibitem{chao2013incentive}
H.~P. Chao and M.~DePillis, ``Incentive effects of paying demand response in
  wholesale electricity markets,'' \emph{Journal of Regulatory Economics},
  vol.~43, no.~3, pp. 265--283, 2013.

\bibitem{vuelvas2017rational}
J.~Vuelvas and F.~Ruiz, ``Rational consumer decisions in a peak time rebate
  program,'' \emph{Electric Power Systems Research}, vol. 143, pp. 533--543,
  2017.

\bibitem{gaming-examples}
\BIBentryALTinterwordspacing
J.~Pierobon. (2013) Two {FERC} settlements illustrate attempts to `game' demand
  response programs. Last accessed 2019-3-30. [Online]. Available:
  \url{https://www.energycentral.com/c/ec/ferc-settlements-illustrate-attempts-game-demand-response-programs}
\BIBentrySTDinterwordspacing

\bibitem{coughlin2008estimating}
K.~Coughlin, M.~A. Piette, C.~Goldman, and S.~Kiliccote, ``Estimating demand
  response load impacts: Evaluation of baseline load models for non-residential
  buildings in california,'' \emph{Lawrence Berkeley National Laboratory},
  2008.

\bibitem{grimm2008evaluating}
C.~Grimm, ``Evaluating baselines for demand response programs,'' in \emph{AEIC
  Load Research Workshop}, 2008.

\bibitem{mathieu2011quantifying}
J.~L. Mathieu, P.~N. Price, S.~Kiliccote, and M.~A. Piette, ``Quantifying
  changes in building electricity use, with application to demand response,''
  \emph{IEEE Transactions on Smart Grid}, vol.~2, no.~3, pp. 507--518, 2011.

\bibitem{wijaya2014bias}
T.~K. Wijaya, M.~Vasirani, and K.~Aberer, ``When bias matters: An economic
  assessment of demand response baselines for residential customers,''
  \emph{IEEE Transactions on Smart Grid}, vol.~5, no.~4, pp. 1755--1763, 2014.

\bibitem{nolan2015challenges}
S.~Nolan and M.~O’Malley, ``Challenges and barriers to demand response
  deployment and evaluation,'' \emph{Applied Energy}, vol. 152, pp. 1--10,
  2015.

\bibitem{weng2015probabilistic}
Y.~Weng and R.~Rajagopal, ``Probabilistic baseline estimation via gaussian
  process,'' in \emph{IEEE Power \& Energy Society General Meeting}, 2015.

\bibitem{zhang2016cluster}
Y.~Zhang, W.~Chen, R.~Xu, and J.~Black, ``A cluster-based method for
  calculating baselines for residential loads,'' \emph{IEEE Transactions on
  smart grid}, vol.~7, no.~5, pp. 2368--2377, 2016.

\bibitem{zhou2016forecast}
X.~Zhou, N.~Yu, W.~Yao, and R.~Johnson, ``Forecast load impact from demand
  response resources,'' in \emph{IEEE Power and Energy Society General
  Meeting}, 2016.

\bibitem{hatton2016statistical}
L.~Hatton, P.~Charpentier, and E.~Matzner-L{\o}ber, ``Statistical estimation of
  the residential baseline,'' \emph{IEEE Transactions on Power Systems},
  vol.~31, no.~3, pp. 1752--1759, 2016.

\bibitem{wang2018synchronous}
F.~Wang, K.~Li, C.~Liu, Z.~Mi, M.~Shafie-Khah, and J.~P. Catal{\~a}o,
  ``Synchronous pattern matching principle-based residential demand response
  baseline estimation: Mechanism analysis and approach description,''
  \emph{IEEE Transactions on Smart Grid}, vol.~9, no.~6, pp. 6972--6985, 2018.

\bibitem{mathieu2011examining}
J.~L. Mathieu, D.~S. Callaway, and S.~Kiliccote, ``Examining uncertainty in
  demand response baseline models and variability in automated responses to
  dynamic pricing,'' in \emph{2011 50th IEEE Conference on Decision and Control
  and European Control Conference}.\hskip 1em plus 0.5em minus 0.4em\relax
  IEEE, 2011, pp. 4332--4339.

\bibitem{muthirayan2017mechanism}
D.~Muthirayan, D.~Kalathil, K.~Poolla, and P.~Varaiya, ``Mechanism design for
  demand response programs,'' \emph{arXiv preprint arXiv:1712.07742}, 2017.

\bibitem{vuelvas2018limiting}
J.~Vuelvas, F.~Ruiz, and G.~Gruosso, ``Limiting gaming opportunities on
  incentive-based demand response programs,'' \emph{Applied Energy}, vol. 225,
  pp. 668--681, 2018.

\bibitem{jacquot2018}
P.~Jacquot, O.~Beaude, S.~Gaubert, and N.~Oudjane, ``Analysis and
  implementation of an hourly billing mechanism for demand response
  management,'' \emph{IEEE Transactions on Smart Grid}, 2018.

\bibitem{muratori2016}
M.~Muratori and G.~Rizzoni, ``Residential demand response: Dynamic energy
  management and time-varying electricity pricing,'' \emph{IEEE Transactions on
  Power systems}, vol.~31, no.~2, pp. 1108--1117, 2016.

\bibitem{yoon2014}
J.~H. Yoon, R.~Baldick, and A.~Novoselac, ``Dynamic demand response controller
  based on real-time retail price for residential buildings,'' \emph{IEEE
  Transactions on Smart Grid}, vol.~5, no.~1, pp. 121--129, 2014.

\bibitem{Krugman2012}
P.~Krugman and R.~Wells, \emph{Microeconomics}, 4th~ed.\hskip 1em plus 0.5em
  minus 0.4em\relax New York: Worth Publishers, 2012.

\bibitem{eia-gov}
\BIBentryALTinterwordspacing
{U.S. Energy Information Administration}. (2018) Electric power monthly. Last
  accessed 2019-3-30. [Online]. Available:
  \url{https://www.eia.gov/electricity/monthly/}
\BIBentrySTDinterwordspacing

\bibitem{king2003predicting}
C.~King and S.~Chatterjee, ``Predicting california demand response,''
  \emph{Public Utilities Fortnightly}, vol. 141, pp. 27--32, 01 2003.

\bibitem{reiss2005household}
P.~C. Reiss and M.~W. White, ``Household electricity demand, revisited,''
  \emph{The Review of Economic Studies}, vol.~72, no.~3, pp. 853--883, 2005.

\bibitem{kiliccote2008installation}
S.~Kiliccote, P.~Gas \emph{et~al.}, ``Installation and commissioning automated
  demand response systems,'' 2008.

\bibitem{hausman2006lmp}
E.~Hausman, R.~Fagan, D.~White, K.~Takahashi, and A.~Napoleon, ``Lmp
  electricity markets: Market operations, market power, and value for
  consumers,'' \emph{Synapse Energy Economics}, 2006.

\bibitem{interconnection2017demand}
P.~Interconnection, ``Demand response strategy,'' \emph{Report. June}, 2017.

\bibitem{doris2011government}
E.~Doris and K.~Peterson, ``Government program briefing: Smart metering,''
  National Renewable Energy Lab.(NREL), Golden, CO (United States), Tech. Rep.,
  2011.

\end{thebibliography}

\appendix

\subsection{Proof of Lemma~\ref{lem:qa-qb-f}}
In order to prove the second statement, we note that
$J^b(f,q,\theta)$ is the sum  of two convex functions $U_1(q) = \pi_0q -
u(q,\theta)$ and $U_2(f, q) = \phi(f-q)$. The minimizer of $U_1$ is
$q^a(\theta)$ and the minimizer of $U_2$ is $q = f$.
Then the minimizer of $J^b = U_1 + U_2$ necessarily lies between the
minimizers of $U_1$ and $U_2$ which implies that $q^b(f,\theta)$ lies between $q^a(\theta)$ and $f$. The first statement
follows from \eqref{eq:opt-con-p-dr-exp}, \eqref{eq:opt-1} and the
properties of the consumer's utility that is monotone increasing.

\subsection{Proof of Lemma \ref{lem:conv-cost}}
We start by showing that $0 \leq \alpha(f,\theta) < 1$,
where $\alpha$ is the cost sensitivity defined as
$\alpha(f,\theta) = \frac{dq^b(f,\theta)}{d f}$.
The optimal consumption $q^b(f, \theta)$ satisfies \eqref{eq:opt-con-p-ndr}.
Holding $\theta$ fixed and differentiating \eqref{eq:opt-con-p-ndr} further we
get,
\begin{align*}
\left(\phi''(f-q) -\frac{\partial^2 u(q,\theta)}{\partial q^2} \right)
\frac{d q^b(f,\theta)}{df} - \phi''(f-q) = 0.
\end{align*}
Convexity of $\phi$ and strict convexity of $-u$ implies the existence of $\frac{d q^b}{df}$ and is given by
\begin{align*}
\alpha(f,\theta) =
\left(\phi''(f-q^b) -\frac{\partial^2 u(q^b,\theta)}{\partial q^2} \right)^{-1} \phi''(f-q^b),
\end{align*}
and satisfies $0 \leq \alpha < 1$. Next, we differentiate $H(f)$ twice to show that $H''(f) > 0$.
Differentiating $H(f)$ we get
\begin{align*}
H'(f) & = (1 - p)\mathds{E}_{\theta} \frac{d J^b(f, q^b, \theta)}{d f} + p \mathds{E}_{\theta}  \frac{d J^c(f, q^c, \theta)}{d f} \nonumber \\
& = (1 - p)\mathds{E}_{\theta} \phi'(f-q^b) -p\pi_2.
\end{align*}
Differentiating once again, we get
\begin{align*}
H''(f) = (1 - p)\mathds{E}_{\theta} \left(1 - \alpha(f,\theta)\right)\phi''(f-q^b).
\end{align*}
Before we showed that $\left(1 - \alpha(f,\theta)\right) > 0$. Then, it follows that $H(f)$ is (strictly) convex if and only if  $\phi$ is (strictly) convex.

\subsection{Proof of Theorem~\ref{lem:f-opt-cond}}
The optimal forecast $f^*$ satisfies the first order condition:
\begin{align*}
H'(f) = (1 - p) \mathds{E}_{\theta}  \frac{d J^b(f, q^b, \theta)}{d f} + p \mathds{E}_{\theta}  \frac{d J^c(f, q^c, \theta)}{d f} = 0.
\end{align*}
The sensitivity of \emph{optimal} cost $J^b(f,q^b,\theta)$ with respect to $f$
is given by
\begin{align*}
\frac{d J^b(f,q^b,\theta)}{d f}
&= \pi_0 \alpha(f,\theta) -\frac{\partial u(q^b, \theta)}{\partial q} \alpha(f,\theta) \nonumber \\
& \quad - \phi'(f-q^b) (\alpha(f,\theta) - 1),
\end{align*}
where $\alpha(f,\theta) = \frac{dq^b(f,\theta)}{d f}$.
Then, taking into account that $q^b(f,\theta)$ satisfies
\eqref{eq:opt-con-p-ndr}, we get
\begin{align}
\frac{d J^b(f, q^b, \theta)}{d f} = \phi'(f-q^b).
\label{eq:der-f-Jb}
\end{align}
The sensitivity of \emph{optimal} cost $J^c(f,q^c,\theta)$ with respect to $f$
is given by
\begin{align*}
\frac{d J^c(f,q^c, \theta)}{d f} = \pi_0 \beta(\theta) - \frac{\partial u(q^c, \theta)}{\partial q} \beta(\theta) - \pi_2 (1 - \beta(\theta)),
\end{align*}
where $\beta(\theta) = \frac{dq^c(\theta)}{d f}$.
As before, $q^c(\theta)$ satisfies \eqref{eq:opt-con-p-dr} and we get
\begin{align}
\frac{d J^c(f,q^c, \theta)}{d f} = -\pi_2 = \pi_0  - \frac{\partial u(q^c, \theta)}{\partial q}.
\label{eq:der-f-Jc}
\end{align}
From equations \eqref{eq:der-f-Jb} and \eqref{eq:der-f-Jc}, we obtain
\begin{align}
\mathds{E}_{\theta} \phi'(f^* - q^b(f^*,\theta)) = \frac{p \pi_2}{1-p}.
\label{eq:foc-f-1}
\end{align}
The optimality condition follows from equations
\eqref{eq:opt-con-p-ndr} and \eqref{eq:opt-con-p-dr} because
\begin{align}
\pi_0 = (1-p)\mathds{E}_{\theta} \frac{\partial u(q^b, \theta)}{\partial q} + p\mathds{E}_{\theta} \frac{\partial u(q^c, \theta)}{\partial q}.
\label{eq:foc-f}
\end{align}
The right hand side in \eqref{eq:foc-f} is the expected marginal utility which
implies that $\pi_0 = M(f^*)$. Since $\phi$ was selected to be convex,
from Lemma~\ref{lem:conv-cost}, $f^*$ is a global minimizer of $H(f)$.
Moreover, if $\phi$ is a strictly convex function, again
from Lemma~\ref{lem:conv-cost}, $f^*$ is unique.

\subsection{Proof of Theorem~\ref{lem:f-opt-novar}}
From equation \eqref{eq:foc-f-1}, we have
\begin{align*}
\mathds{E}_{\theta} \phi'(f^* - q^b(f^*,\theta)) =
	\mathds{E}_{\theta} 1/\lambda(f^* - q^b(f^*,\theta)) = \frac{p \pi_2}{1-p}.
\end{align*}
This implies,
\begin{align*}
\lim_{p \rightarrow 0} f^* - \mathds{E}_{\theta} q^b = \lim_{p \rightarrow 0}  \mathds{E}_{\theta} \delta \tilde{f}^*(\theta) = 0.
\end{align*}
From the optimality condition for $q^b(f^*,\theta)$ \eqref{eq:opt-con-p-ndr} and when $\frac{\partial u^2(q,\theta)}{\partial q^2} = -1/d$,
\begin{align*}
q^b(f^*,\theta) = q^a(\theta) + d/\lambda(f^* - q^b(f^*,\theta)).
\end{align*}
Taking expectations on both sides we get
\begin{align*}
\lim_{p \rightarrow 0} \mathds{E} q^b(f^*,\theta) = \mathds{E} q^a(\theta).
\end{align*}
That is,
\begin{align*}
\lim_{p \rightarrow 0} f^* - \mathds{E}_{\theta} q^a = \lim_{p \rightarrow 0}  \mathds{E}_{\theta} \delta f^*(\theta) = 0.
\end{align*}

\subsection{Proof of Theorem \ref{th:f-opt-ub}}
The consumptions $q^a(\theta)$ and $q^c(\theta)$ have the expressions:
\begin{align*}
q^a(\theta) &= \mu^{-1}(\pi_0,\theta), \\
q^c(\theta) &= \mu^{-1}(\pi_0 + \pi_2,\theta).
\end{align*}
Also from \eqref{eq:foc-f}, we get
\begin{align*}
\pi_0 & =
(1-p)\mathbb E_{\theta} \frac{\partial u(q^b,\theta)}{\partial q} +
p\mathbb E_{\theta} \frac{\partial u(q^c,\theta)}{\partial q} \nonumber \\
& = (1-p) \mathbb E_{\theta} \frac{\partial u(q^b,\theta)}{\partial q} +
    p(\pi_0 + \pi_2).
\end{align*}
This implies,
\begin{align*}
\mathbb E_{\theta} \mu(q^b(\theta),\theta) =
\pi_0 - \frac{p}{1-p}\pi_2.
\end{align*}

Since the utility function $u$ is quadratic in $q$, the marginal utility
$\mu = u'$ is affine in $q$. Moreover, since $\mu'$ is independent of
the random variable $\theta$, it holds
\begin{align*}
\mathbb E_{\theta}
\mu^{-1}(\mathbb E_{\theta} \mu (q(\theta), \theta),\theta) =
\mathbb E_{\theta} q(\theta),
\end{align*}
and substituting $q^b(\theta)$ in the previous expression, we obtain
\begin{align*}
\mathbb E_{\theta} q^b(\theta) =
\mathbb E_{\theta}
\mu^{-1}\left( \pi_0 - \frac{p}{1-p}\pi_2, \theta  \right).
\end{align*}

The consumer's utility $u$ and the penalty
$\phi$ are quadratic functions, then their derivatives $u'$ and
$\phi'$ are affine and their inverse functions are also affine.
Moreover, the derivative of the inverse functions satisfy:
\begin{align*}
\frac{\partial^k}{\partial x^k} \mu^{-1}(x) & =
\left\{
\begin{array}{ll}
d, & \text{if } k=1, \\
0, & \text{if } k>1.
\end{array}
\right. \nonumber \\
\frac{d^k}{d x^k} \phi'^{-1}(x) &=
\left\{
\begin{array}{ll}
\lambda, & \text{if } k=1, \\
0, & \text{if } k>1.
\end{array}
\right.
\end{align*}

Then, the expressions of
$\mu^{-1}\left(\pi_0 - \frac{p\pi_2}{1-p},\theta\right)$ and
$\phi'^{-1}\left(\frac{p\pi_2}{1-p}\right)$
become
\begin{align*}
\mu^{-1}\left(\pi_0 - \frac{p\pi_2}{1-p},\theta\right)
&=
\mu^{-1}(\pi_0,\theta) + d \frac{p\pi_2}{1-p} \\
&=
q^a(\theta) + d \frac{p\pi_2}{1-p},
\end{align*}
and
\begin{align*}
\phi'^{-1}\left(\frac{p\pi_2}{1-p}\right) =
\lambda \frac{p\pi_2}{1-p}.
\end{align*}

Using \eqref{eq:foc-f-1}, we get
\begin{align}
f^* = \mathbb E_{\theta} q^b(\theta) +
	\phi'^{-1}\left(\frac{p}{1-p}\pi_2\right).
	\label{eq:srepc}
\end{align}

By taking expectations,
\begin{align}
	\mathbb E_{\theta} q^b(\theta) =
	\mathbb E_{\theta} q^a(\theta) + d \frac{p\pi_2}{1-p}.
	\label{eq:condeqb}
\end{align}
and by substitution in \eqref{eq:srepc}, we obtain
\begin{align*}
f^* = \mathds E_{\theta} q^a(\theta) +
\left( d + \lambda \right)
\frac{p\pi_2}{1-p}.
\end{align*}

\subsection{Proof of Theorem~\ref{th:f-opt-ub-db}}

Define a family of penalty functions with deadband as follows:
\[
\phi_{\Delta}(x) =
\left\{
\begin{array}{ll}
	0, & \mathrm{if}\ |x| < \epsilon-\Delta, \\
	(|x|+\Delta-\epsilon)^3/(6\lambda\Delta),
	& \text{if}\ \epsilon-\Delta \leq |x| \leq \epsilon, \\
	\Delta^2/(6\lambda) + (|x|-\epsilon) \Delta/(2\lambda) \\
	+ (|x| - \epsilon)^2/(2\lambda),
	& \mathrm{if}\ |x| > \epsilon,
\end{array}
\right.
\]
for $0 \leq \Delta < \epsilon$. Note that $\phi_{\Delta}$ is continuous with
continuous derivatives up to second order for $0 < \Delta <\epsilon$,
and it approaches the penalty function $\phi$ given by \eqref{eq:penalty-db}
as $\Delta$ approaches zero, \emph{i.e.}
$\phi = \lim_{\Delta\rightarrow 0^+} \phi_{\Delta}$.

The derivative of $\phi_{\Delta}$ is:
\[
\phi_{\Delta}'(x) =
\left\{
\begin{array}{ll}
	0, & \mathrm{if}\ 0 \leq x < \epsilon-\Delta, \\
	(x+\Delta-\epsilon)^2/(2\lambda\Delta),
	& \text{if}\ \epsilon-\Delta \leq x \leq \epsilon, \\
	\Delta/(2\lambda) + (x - \epsilon)/\lambda,
	& \mathrm{if}\ x > \epsilon.
\end{array}
\right.
\]
for $x \geq 0$ and $\phi_{\Delta}'(x)=-\phi_{\Delta}'(-x)$ for $x \leq 0$,
which is invertible for any $x \neq 0$.

From the fact that this penalty function is double differentiable, the optimality conditions established before hold for this specific case as well. We do a case based analysis.

{\em Case $f^* \geq \mathds{E}_{\theta} q^a(\theta) + \epsilon$:} From Lemma \ref{lem:qa-qb-f} we have that $f^* \geq q^b(f^*,\theta) \ \forall \theta$. Then using \eqref{eq:foc-f-1}, the convexity of $\phi'$ for $x \geq 0$, that $f^* \geq q^b(f^*,\theta) \ \forall \theta$ and Jensen's inequality, we get
\begin{align*}
\phi'(\mathbb{E} f^* - q^b(f^*,\theta))  \leq \mathds{E}_{\theta} \phi'(f^* - q^b(f^*,\theta)) = \frac{p \pi_2}{1-p}.
\end{align*}

It is always possible to choose $\Delta$ such that
\[
	0 < \frac{\Delta}{2\lambda} < \frac{p\pi_2}{1-p}.
\]

For this value of $\Delta$, $\phi_{\Delta}'(x)$ is always restricted
to $x > \epsilon$.
Since $\phi'_{\Delta}$ is affine and invertible, we get
\[
	\phi_{\Delta}'^{-1}
	\left( \frac{p\pi_2}{1-p} \right) =
	\lambda \frac{p\pi_2}{1-p} - \frac{\Delta}{2} + \epsilon.
\]

Since $\mathbb{E} (f^* - q^b(f^*,\theta)) \geq 0$, the fact that $\phi'$ is increasing for $x\geq 0 $ and from the previous equation it follows that

\begin{align*}
f^* \leq \mathbb E_{\theta} q^b(\theta) +
	\phi'^{-1}_{\Delta} \left(\frac{p}{1-p}\pi_2\right),
\end{align*}
and substituting the value of $\mathbb E_{\theta} q^b(\theta)$ given
by \eqref{eq:condeqb}, we obtain
\begin{align*}
f^* \leq \mathbb E_{\theta} q^a(\theta) + \left( d + \lambda \right)
\frac{p\pi_2}{1-p} - \frac{\Delta}{2} + \epsilon.
\end{align*}

Hence, the result for the penalty function with deadband
$\phi$ defined in \eqref{eq:penalty-db} is obtained by taking
limit when $\Delta$ approaches zero,
\begin{align*}
f^* \leq \mathds E_{\theta} q^a(\theta) +
\left( d + \lambda \right)
\frac{p\pi_2}{1-p} + \epsilon.
\end{align*}

{\em Case $f^* < \mathds{E}_{\theta} q^a(\theta) + \epsilon$:} By this case it follows that
\begin{align*}
f^* \leq \mathds E_{\theta} q^a(\theta) +
\left( d + \lambda \right)
\frac{p\pi_2}{1-p} + \epsilon.
\end{align*}

{\em Individual Rationality}: For an $\epsilon$ such that  $\max\{q_{\max} - \mathds{E}_\theta q^a(\theta), \mathds{E}_\theta q^a(\theta) - q_{\min} \} \leq \epsilon$, the report $f = \mathds{E}_\theta q^a(\theta)$ is individually rational. Thus the optimal baseline report $f^*$ should be individually rational. Hence proved.

\subsection{Proof of Lemma~\ref{lem:caiso_adj}}
We start by showing that $f^c > f^-$. Recall that $f^-$ is the average of
consumption on the days prior to the previous $m$ non-event days. The
consumption on these days only appear in the denominator of the CAISO's
baseline estimate for future DR events. Hence, the incentive for the consumer is
to reduce the consumption on these days so as to inflate the baseline. On the
other hand, $f^c$ is the average of the consumption on the previous $m$
non-event days. And the consumption on these days only appear in the numerator
of the baseline estimate for any future DR events, through the term $f^c$.
Hence, the incentive for the consumer is to increase the consumption on these
days. Since everything else is the same for the day prior to the non-event
day and the non-event day except for this incentive to reduce and increase,
respectively, we conclude that $f^c > f^-$. This implies:
\begin{align*}
q^{-*}(\theta) > \mu^{-1}(\pi_0-\pi_2,\theta).
\end{align*}
Hence,
\begin{align*}
\bar f^c - q_a(\theta)
&>
\mu^{-1}(\pi_0-\pi_2,\theta) -
\mu^{-1}(\pi_0,\theta) \\
&= d\pi_2.
\end{align*}
Taking expectation with respect to $\theta$ we obtain
\begin{align*}
\mathds E_{\theta} (\bar f^c - q_a(\theta) ) > d\pi_2,
\end{align*}
and this completes the proof.
\end{document}